\documentclass[twoside,leqno,twocolumn]{article}

\usepackage{hyperref}

\usepackage{tabularx}
\usepackage{booktabs}
\usepackage{subcaption}
\usepackage{todonotes}
\usepackage{xspace}
\usepackage{amsmath}
\usepackage{amssymb}
\usepackage{amsthm}
\usepackage{mathtools}
\usepackage{fnpct}
\usepackage{siunitx}

\newtheorem{theorem}{Theorem}[section]
\newtheorem{lemma}{Lemma}[section]

\newtheorem{Definition}{Definition}[section]

\DeclareMathOperator{\len}{len}
\DeclareMathOperator{\cost}{cost}

\DeclarePairedDelimiterX\set[1]{\lbrace}{\rbrace}{#1}
\DeclarePairedDelimiter{\floor}{\lfloor}{\rfloor}
\DeclarePairedDelimiter{\ceil}{\lceil}{\rceil}

\newcommand{\Natural}{\mathbb{N}}

\newcommand{\low}{\ell}
\newcommand{\high}{u}
\newcommand{\layer}{\lambda}

\newcommand{\setz}[1]{\set{0, \ldots, #1}}
\newcommand{\dotcup}{\mathbin{\dot{\cup}}}

\newcommand{\Sofaclap}{\textsc{Solar Farm Cable Layout Problem}\xspace}
\newcommand{\sofaclap}{\textsc{SoFaCLaP}\xspace}
\newcommand{\sofaclapfeas}{\textsc{Constrained Layer Tree}\xspace}
\newcommand{\sofaclapembed}{\textsc{SoFaCLaP Embedding}\xspace}
\newcommand{\threepartition}{\textsc{3-Partition}\xspace}

\newtheorem{observation}[lemma]{Observation}

\begin{document}

\newcommand\relatedversion{}

\title{The Constrained Layer Tree Problem and Applications to Solar Farm Cabling}
\author{Thomas Bläsius\thanks{Karlsruhe Institute of Technology}
\and Max Göttlicher\thanks{Karlsruhe Institute of Technology, funded by German Research Foundation (DFG) as part of the Research Training Group GRK 2153: Energy Status Data -- Informatics Methods for its Collection, Analysis and Exploitation.}
\and Sascha Gritzbach\footnotemark[1]
\and Wendy Yi\footnotemark[2]}

\date{}

\maketitle

\begin{abstract} \small\baselineskip=9pt%
  Motivated by the cabling of solar farms, we study the problem \sofaclapfeas.
  At its core, it asks whether there exists a tree that connects a set of sources (the leaves) to one sink (the root) such that certain capacity constraints at the inner nodes are satisfied.
  Our main algorithmic contribution is a dynamic program with various optimizations for \sofaclapfeas.
  It outperforms the previously used MILP by multiple orders of magnitude.
  Moreover, our experiments show that the somewhat abstract problem \sofaclapfeas is actually the core of the cabling problem in solar farms, i.e., the feasible solution produced by our dynamic program can be used to bootstrap an MILP that can then find good solutions for the original cabling problem efficiently.
\end{abstract}

\section{Introduction}
A solar farm consists of various different components.
Besides the prominent PV\footnote{PV -- short for photovoltaic} panels, which are mounted on racks in groups of so-called PV strings, there are less visible components that are vital to harvest the produced energy.
To give an example, the PV strings can be paired up using Y-connectors.
Multiple Y-connectors are connected to combiner boxes, which in turn are grouped by recombiner boxes.
These recombiner boxes are then connected to an inverter, which transforms the direct current produced by the PV strings into alternating current.
Finally, step-up transformers feed the generated electricity into the power grid.
Besides the hierarchical combination of PV strings, these components can come with varying monitoring and safety equipment.
Beyond this example, there is a multitude of components to choose from and which components exactly are used heavily depends on the decisions made in the planning process.
See \cite{abb2019_photovoltaic_plants} for more technical details.
A common pattern, however, is that PV strings, forming the basic unit, are connected in a hierarchical fashion, with components allowing for increasing current in the higher layers of the hierarchy.

When considered as an algorithmic problem, the task boils down to constructing a tree with a fixed number of layers that has the PV strings as leaves.
As a hard-constraint, each component on a given layer can only handle current within a certain range, i.e., there are capacities that restrict the number of strings in the subtree below the component.
The goal is to choose and place components on given candidate positions such that the total cabling cost is minimized.
This problem has been formalized by Gritzbach, Stampa, and Wolf as the so-called \Sofaclap (\sofaclap)~\cite{gritzbach2022sofaclap}.

More generally, this can be categorized as a network design problem~\cite{Networ_Desig_Probl-Ciesl09}.
The arguably oldest and most famous network design problem is the minimum spanning tree (MST) problem~\cite{Histor_Minim_Spann_Tree_Probl-GrahamHell85}.
When it comes to practical applications, however, just computing the MST is often not sufficient.
Unfortunately, seemingly inconspicuous additional requirements or variations of the problem often make it hard.
In fact, the book by Garey and Johnson~\cite{Comput_Intrac-GareyJohns90} contains a whole section on spanning trees listing thirteen different NP-hard variants of the problem.
This includes the Steiner tree problem~\cite{Stein-Hakim71}, where the resulting tree only needs to contain a specified set of terminal vertices, while all other vertices, called Steiner points, are optional~\cite{Stein-Hakim71}.
Although the Steiner tree problem and its various variants are computationally much harder than the MST problem, efficient MST-algorithms are often a core ingredient for (heuristically) solving the more complicated problem; see, e.g., \cite{Solvin_Stein-Ljubić21} for a recent survey on variants of the Steiner tree problem.

As pointed out by Gritzbach et al.~\cite{gritzbach2022sofaclap}, \sofaclap is closely related to the Steiner tree problem.
There are, however, differences coming from the hierarchical (or multi-layer) nature of the desired trees as well as from the capacity constraints of the connecting components.
Moreover, there are resemblances to other classes of problems that have been studied before, like multi-layer network design~\cite{Multi_Layer_Networ_Desig_Probl-KnippLardeux07,Taxon_Multil_Networ_Desig_Survey-Crain22}, hierarchical network design~\cite{Hierar_Networ_Desig_Probl-Curren86}, and multi-level facility location~\cite{Multi_Level_Facil_Locat_Probl-Ortiz18}.
These problems are often directly motivated by applications, which results in a multitude of variants with varying constraints and different optimization criteria.
It is not surprising that these problems are predominantly studied by the operations research community, where mathematical programming formulations, often mixed integer linear programs (MILP), can be flexibly adjusted to the specific needs and years of engineering can be utilized by using off-the-shelf solvers like Gurobi \cite{gurobi}.
For \sofaclap, Gritzbach et al.~\cite{gritzbach2022sofaclap} also provide and evaluate a formulation as MILP.
They observe that, while Gurobi is able to solve small and medium sized instances optimally in reasonable time, it struggled on several larger instances to even find a feasible solution satisfying all hard-constraints.
Moreover, they observe that, if Gurobi can find a feasible solution, it is also quite good at finding a solution of good quality.
This indicates that the core problem here does not come from the optimization problem of finding a cheap layout but from the decision problem of finding any tree that satisfies the capacity constraints with the given set of components.

In this paper, we study a problem that we call \sofaclapfeas, which captures the difficult core of the solar farm problem in isolation.
Our main contributions are as follows.
We provide a dynamic program that solves \sofaclapfeas and devise a multitude of optimizations that improve its performance.
Our evaluation shows that our dynamic program heavily outperforms Gurobi trying to solve the MILP formulation\footnote{This is the same formulation as for \sofaclap but with a constant objective function since we only ask for feasibility.} of \sofaclapfeas.
More specifically, our dynamic program is able to solve a multitude of instances, where Gurobi fails to find a feasible solution.
Moreover, bootstrapping the MILP solver with a feasible solution produced by our dynamic program lets it quickly find good solutions for instances that were unsolvable before, which shows that \sofaclapfeas is indeed the difficult core of \sofaclap.
For the optimizing of a feasible solution for \sofaclap, we also provide and evaluate different heuristic solutions.

\subsection{The \sofaclapfeas Problem}

In a nutshell, \sofaclapfeas can be described as follows; see Section~\ref{sec:prelims} for a formal definition.
The goal is to build a tree with $\lambda$ layers and $n_0$ leaves on the bottom layer that satisfies the following constraints.
For each layer $i$, there is an upper bound $n_i$ on the number of vertices in layer $i$.
Additionally, there are lower and upper capacities $\ell_i$ and $u_i$, and any vertex on layer $i$ must have between $\ell_i$ and $u_i$ leaves below it.

Relating this back to the solar farm problem, $n_0$ is the number of PV strings we want to connect, the layers represent the different types of components, $n_i$ for $i > 0$ gives an upper bound on the number of available components of each type, and the capacities specify how many strings can be bundled by each of the components.
Although the \sofaclapfeas problem is clearly motivated by the cabling of solar farms, it is a quite natural problem and we believe that it can be of independent interest as a subproblem appearing in other applications\footnote{We named it \sofaclapfeas, to reflect the fact that it is not specific to the solar farm application.}.

We note that the \sofaclapfeas problem is not completely new.
Gritzbach et al.\ have mentioned it as a subproblem of \sofaclap and in fact conjectured it to be solvable in polynomial time~\cite[Conjecture~1]{gritzbach2022sofaclap}.
To the best of our knowledge, \sofaclapfeas has not been studied beyond that.

\subsection{Our Dynamic Program}

We give a dynamic program for \sofaclapfeas in Section~\ref{sec:dp-algor-sofacl} that runs in $\mathcal O(n_0^{\lambda + 2}\cdot \lambda^2\cdot \log n_0)$ time.
We note that this partially confirms the conjecture by Gritzbach et al.~\cite[Conjecture~1]{gritzbach2022sofaclap}; with some major caveats.
Firstly, this is only polynomial in the number of leaves $n_0$ if the number of layers $\lambda$ is constant.
Secondly, an instance of \sofaclapfeas can be represented very compactly.
One only needs to specify $n_0, \dots, n_\lambda$ as well as the capacities $\ell_1, \dots, \ell_\lambda$ and $u_1, \dots, u_\lambda$, which are just $3 \lambda + 1$ numbers.
Thus, unless we use a unary encoding for the numbers, even a running time of $\mathcal O(n_0)$ is exponential in the input size.
We note that, with this perspective, it is even unclear whether the problem is in NP, as explicitly writing down a solution uses an exponential number of bits compared to the number of bits required to encode the instance.

Our dynamic program goes over the number of leaves in the solution.
Roughly speaking, for each number of leaves up to $n_0$, we compute all possible ways a tree with that many leaves can be built.
As typical for dynamic programs, the number of these partial solutions is kept low by viewing different trees that share certain characteristics as equivalent.
Additionally, we provide several techniques in Section~\ref{sec:optim-dp-algor} that heavily improve the practical performance of our dynamic program while not affecting its correctness.
For example, one type of optimizations reduces the number of stored partial solutions significantly by identifying partial solutions that do not contribute to the full solution.
Other optimizations try to complete an already computed partial solution greedily to a full solution, and terminate the dynamic program early on success.

\subsection{Performance}

In our experiments on randomly generated instances of \sofaclapfeas, our dynamic program with the proposed optimizations solves a significant number of instances that are not solved by Gurobi in reasonable time.
On more than $90$\% of the tested instances, our dynamic program achieves a speed-up of more than $100$ compared to Gurobi and often even reaches a much higher speed-up\footnote{In fact, there are multiple instances that we solve in \SI{2}{ms} for which Gurobi runs in the timeout of \SI{1}{h}; a speed-up of six orders of magnitude.}.
Besides this comparison, we make a detailed evaluation of our dynamic program that yields interesting insights into how and why it performs the way it does.

\subsection{Optimizing \sofaclap}

Our study of the \sofaclapfeas problem is motivated by the application of building solar farms.
To verify that solving \sofaclapfeas can actually help for this kind of application, we run Gurobi initialized with a valid solution to evaluate whether it can find good solutions for instances of \sofaclap for which it was not able to find any solution before.
Our experiments show that this is indeed true, confirming that \sofaclapfeas is the hard core of the problem.
Moreover, we provide and evaluate heuristics as an alternative to an MILP for solving \sofaclap.
Our heuristics yield worse solutions than Gurobi but are slightly faster.

\section{Preliminaries}
\label{sec:prelims}

Let $G = (V, E)$ be a directed acyclic graph.
Vertices with no incoming or outgoing edge are called \emph{source} and \emph{sinks}, respectively.
We call $G$ a \emph{layer graph} if every path from a source to a sink has the same length $\lambda$.
This naturally partitions the vertices into $\lambda + 1$ disjoint \emph{layers} $V = V_0 \dotcup \dots \dotcup V_\lambda$ such that every edge goes from a vertex in $V_{i - 1}$ to a vertex in $V_i$ for some $i \in [\lambda]$, where $[\lambda]$ denotes the set $\{1, \dots, \lambda\}$.
Note that $V_0$ are the sources and $V_\lambda$ are the sinks.
We denote with $n_i = |V_i|$ the number of vertices in layer $i$.
A layer graph is \emph{complete} if all possible edges exist, i.e., $E = \bigcup_{i \in [\lambda]}V_{i - 1} \times V_{i}$.

A \emph{layer tree} is a layer graph that is also a rooted tree, i.e., the leaves (sources) are in layer $0$, there is exactly one root (sink) in layer $\lambda$, the distance from every leaf to the root is $\lambda$, and each vertex except the root has exactly one parent.
The \emph{weight} of a vertex $v$ in a layer tree is the number of leaves in the subtree below $v$, i.e., the number of leaves that have a directed path to $v$.
We denote the weight with $w(v)$.

The input of the \sofaclapfeas problem is a complete layer graph $G$ with layers $0, \dots, \lambda$ together with a pair of \emph{lower} and \emph{upper capacity} $(\ell_i,u_i)$ for each $i \in [\lambda]$.
The problem \sofaclapfeas then asks whether $G$ has a layer tree containing all $n_0$ sinks as subgraph such that the weight of each vertex lies between the lower and upper capacity in the corresponding layer, i.e., for $v \in V_i$ it holds that $\ell_i \le w(v) \le u_i$.
We call such a layer tree \emph{valid}.

We note that an instance of \sofaclapfeas is completely determined by specifying the numbers $n_0$ and $(n_i, \ell_i, u_i)_{i \in [\lambda]}$.
Although layer $0$ has no capacities, it is sometimes convenient to assume that $\ell_0 = u_0 = 1$.

\paragraph{Multiple Sinks}

When coming from the solar farm application, it often makes sense to have not just one but multiple transformers as sinks.
In this case, the solution should not be a tree but a forest.
We note that this case is already covered by our definition of \sofaclapfeas, as one can always insert an additional top layer with just one vertex.
Thus, the requirement of having just one sink is without loss of generality.

\paragraph{Normalization}

We call an instance of \sofaclapfeas \emph{normalized} if, with increasing layer, the number of vertices is decreasing and the capacities are increasing, i.e., $n_i \le n_{i - 1}$, $\ell_i \ge \ell_{i - 1}$, and $u_i \ge u_{i - 1}$ for $i \in [\lambda]$.
It is easy to verify that any instance can be transformed into an equivalent normalized instance in $\mathcal{O}(\lambda)$.
We assume for the remainder of the paper that all instances are normalized.
\begin{lemma}
  \label{lemma:normalization}
  Any instance of \sofaclapfeas can be transformed into an equivalent normalized instance in $\mathcal{O}(\lambda)$.
\end{lemma}

\begin{proof}
  Given an instance determined by the values $n_0$ and $(n_i, \ell_i, u_i)_{i \in [\lambda]}$, we construct an equivalent normalized instance $(\tilde{n}_i, \tilde{\ell}_i, \tilde{u}_i)_{i \in [\lambda]}$ as follows.

  For every layer $i \in [\lambda]$, set $\tilde{n}_i = \min(n_i, \tilde{n}_{i - 1})$, $\tilde{\low}_i = \max(\low_i, \tilde{\low}_{i - 1})$, and $\tilde{\high}_{i - 1} = \min(\high_{i - 1}, \tilde{\high}_{i})$.
  Clearly, the resulting instance is normalized.

  Going from the original to the normalized instance, the number of vertices in non-source layers is only decreased and the capacity intervals are only made smaller.
  Thus, if the original instance was a no-instance, it remains a no-instances after the normalization.
  Moreover, any solution of the original instance remains a solution of the normalized instance: In the solution, the number of vertices in each layer decreases with increasing layers.
  Moreover, the weight of any vertex is at least the weight of its children.
  Thus, increasing the lower bound of the parent to the lower bound of the child and decreasing the upper bound of the child to the upper bound of the parent keeps the weight in the solution between the upper and lower bound.\hfill
\end{proof}

\paragraph{Solar Farm Cabling}

With the above formalization, we can also define the \sofaclap problem studied by Gritzbach et al.\ \cite{gritzbach2022sofaclap} as follows.
As for the \sofaclapfeas problem, the input consists of a complete layer graph $G = (V, E)$ together with the capacities.
Additionally, every edge of $G$ has a \emph{length} $\len: E \to \mathbb R$ indicating the distance between the corresponding components in the solar farm\footnote{The $\len$ function is usually compactly represented by specifying geometric positions of the vertices.}.
Moreover, there is a \emph{cable cost function} $\cost\colon \mathbb N \to \mathbb R$ that indicates how expensive the required cable is for a given load, i.e., if the vertex $u$ in a layer tree has weight $w(u)$, then the edge to its parent costs $\cost(w(u))$ per unit of length.
Thus, the \emph{cost} of the edge $e = (u, v)$ is $\cost(w(u)) \cdot \len(e)$.
The goal of \sofaclap is to find a valid layer tree with minimum total edge cost.

\section{DP Algorithm for \sofaclapfeas}
\label{sec:dp-algor-sofacl}

In this section, we first describe a basic dynamic program that solves \sofaclapfeas.
We propose further optimizations in Section~\ref{sec:optim-dp-algor}.

The main idea of the DP is to obtain valid layer trees with a higher number of leaves by combining two valid layer trees with fewer leaves.
However, to obtain a solution for a given instance of \sofaclapfeas, it is not necessary to actually compute all possible trees explicitly, and much less information is needed to be able to reconstruct such a solution.
In fact, it is sufficient to know how many vertices are in each layer, regardless of the actual topology.
Thus, we consider two layer trees to be equivalent if they have the same number of vertices in each layer.
More formally, let $a \coloneqq (a_0, \ldots, a_\layer) \in \mathbb{N}^{\layer + 1}$ be a vector where the number $a_i$ describes the number of vertices in layer $i$.
We call such a vector $a$ a \emph{partial solution} with $a_0$ leaves if there exists a valid layer tree with $a_i$ vertices in layer $i \in \setz{\layer}$.
Observe that, by definition, there is a partial solution with $a_0 = n_0$ leaves if and only if there is a valid layer tree that connects all sources, i.e., if the given instance is a yes-instance.

\subsection{Ignoring Lower Capacities}

For now, we assume that the given instance of \sofaclapfeas has no lower capacities, i.e., $\low_i = 0$ for every layer $i \in \setz{\layer}$.
This simplifies the description of the dynamic program, and we will later see that it can be slightly adapted to also work with lower capacities.

Once we know all partial solutions with fewer leaves, we can combine these to obtain all partial solutions for more leaves.
Combining two partial solutions is roughly the sum of the number of vertices in each layer, with the exception that in a layer, where both solutions only have one single vertex, it is possible to merge these two vertices into one.
As such a merging of two vertices is optional, we may obtain multiple new partial solutions by combining two partial solutions.

More formally, we define the combination of two partial solutions as follows.
Given a partial solution $a = (a_0, \ldots, a_\layer)$, we refer to the highest layer $i$ with $a_i > 1$ as the \emph{branching layer} of $a$.
Observe that in a corresponding tree, the branching layer consists of multiple vertices, and the tree forms a path above the branching layer.

Let $a = (a_0, \ldots, a_\layer)$ and $b = (b_0, \ldots, b_\layer)$ be partial solutions, and let $k_a$ and $k_b$ be the branching layer of $a$ and $b$, respectively.
For some $k \in \set{\max\{k_a, k_b\}, \ldots, \layer}$, the \emph{$k$-combination} of $a$ and $b$ is the vector $(a_0 + b_0, \ldots, a_{k} + b_{k}, 1, \ldots, 1)$.
Note that the branching layer of a $k$-combination is layer $k$ by definition.

Observe that combining partial solutions has a correspondence on trees.
Let $T_a$ and $T_b$ be layer trees with branching layers $k_a$ and $k_b$, respectively.
For some $k \in \set{\max\set{k_a, k_b}, \ldots, \layer}$, we define the $k$-combination of $T_a$ and $T_b$ as follows.
We take the disjoint union of $T_a$ and $T_b$, and for every $i \in \set{k + 1, \ldots, \layer}$, we contract the two vertices in layer $i$ into one.
Figure~\ref{fig:combine} shows such a $k$-combination.
\begin{figure}
  \centering
  \includegraphics{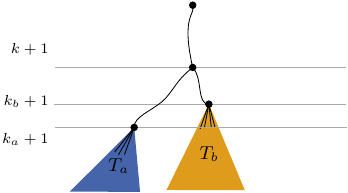}
  \caption{The $k$-combination of two trees $T_a$ and $T_b$ with branching layer $k_a$ and $k_b$, respectively.}
  \label{fig:combine}
\end{figure}
For $k$-combinations of layer trees, we observe that the weight of vertices only change in layers above $k$.

\begin{observation}
  \label{observation:equal-load}
  Let $T$ be the $k$-combination of two layer trees $T_a$ and $T_b$.
  For every vertex in some layer at most $k$ in $T$, the weight is equal to its corresponding vertex in $T_a$ or $T_b$.
\end{observation}
With this, we can show the following lemma.
It essentially states that it can be easily checked which combinations of partial solutions are partial solutions themselves.
\begin{lemma}
  \label{lemma:dp-correct}
  Let $a$ and $b$ be partial solutions with $a_0$ and $b_0$ leaves, respectively.
  The $k$-combination $c = (c_0, \ldots, c_\layer)$ of $a$ and $b$ for some $k$ yields a partial solution with $a_0 + b_0$ leaves if and only if
  \begin{itemize}
    \item $c_i$ is at most $n_i$ for every $i \in \setz{\layer}$, and
    \item $c_0 \leq \high_{k + 1}$.
  \end{itemize}
\end{lemma}
\begin{proof}
  Let $a = (a_0, \ldots, a_\layer)$ and $b = (b_0, \ldots, b_\layer)$.
  As $c$ is their $k$-combination, it holds that $c_i = a_i + b_i$ if $i \leq k$ and $c_i = 1$ if $i > k$.

  If $c_i > n_i$, then $c$ is not a partial solution by definition.
  The same is true if there is a layer that consists of a single vertex whose upper capacity is not large enough to connect all leaves.

  Now, assume $c$ satisfies the conditions from the lemma, i.e., $c_i \leq n_i$ for all $i \in \setz{\layer}$ and $c_0 \leq \high_{k + 1}$.
  We prove that $c$ is a partial solution.
  As $a$ and $b$ are partial solutions, there are two valid layer trees $T_a$ and $T_b$ that correspond to $a$ and $b$, respectively.
  Consider the $k$-combination $T_c$ of $T_a$ and $T_b$.
  By definition, $T_c$ corresponds to the vector $c$, and each layer $i \in \setz{\layer}$ in $T_c$ consists of at mosts $n_i$ vertices.
  For the capacity constraint, we consider the layers at most $k$ and the layers above $k$ separately.
  In the layers at most $k$, the weight of a vertex in $T_c$ is equal to the weight of its corresponding vertex in $T_a$ or $T_b$ by Observation~\ref{observation:equal-load}.
  Since $T_a$ and $T_b$ are valid, the capacity constraints in the layers at most $k$ in $T_c$ are satisfied.
  In every layer above $k$, $T_c$ consists of a single vertex and all leaves in $T_c$ are connected to this vertex.
  As we assumed that the upper capacity of a higher layers is at least the upper capacity of a lower layer, the capacity constraint is satsified for all layers above $k$.
  Thus, $T_c$ is a valid layer tree and by that, $c$ a partial solution which concludes the proof.\hfill
\end{proof}
Conversely, we prove that we can obtain any partial solution by combining partial solutions for fewer leaves.
\begin{lemma}
  \label{lemma:dp-complete}
  Let $c \in \mathbb{N}^{\layer + 1}$ be a partial solution with $c_0 > 1$ leaves.
  There exist two partial solutions $a, b \in \mathbb{N}^{\layer + 1}$ with $a_0$ and $b_0$ leaves, respectively, with $a_0 < c_0$ and $b_0 < c_0$ such that $c$ is the $k$-combination of $a$ and $b$ for some $k$.
\end{lemma}
\begin{proof}
  Let $c = (c_0, \ldots, c_\layer)$ be a partial solution for more than one leaf, and $T_c$ a corresponding valid layer tree.
  Further, let $k$ be the branching layer of $c$.
  We construct two layer trees $T_a$ and $T_b$ whose $k$-combination gives us $T_c$ as follows.
  Let $u$ be the vertex in layer $k + 1$ of $T_c$.
  As it is above the branching layer, $u$ has at least two children.
  Let $v$ be a child of $u$ in layer $k$.
  We obtain the layer tree $T_a$ by deleting $v$ and the subtree rooted in $v$ from $T_c$.
  Similarly, $T_b$ is constructed by deleting all children of $u$ except for $v$ and their subtrees in $T_c$.
  An example can be seen in Figure~\ref{fig:partition-tree}.
  \begin{figure}
    \centering
    \includegraphics{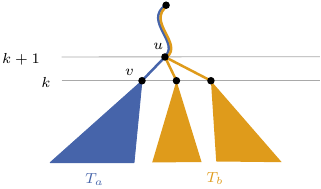}
    \caption{Any valid tree $T_c$ is a $k$-combination of two valid trees $T_a$ and $T_b$.}
    \label{fig:partition-tree}
  \end{figure}

  Observe that $c$ is the $k$-combination of $a$ and $b$ where $a$ and $b$ are the corresponding vectors to $T_a$ and $T_b$, respectively.
  As both $T_a$ and $T_b$ are subtrees of $T_c$, they both have at most as many vertices as $T_c$ in each layer, and the weight of each vertex is at most the weight of its corresponding vertex in $T_c$.
  Thus, we can obtain $c$ by combining two partial solutions $a$ and $b$.\hfill
\end{proof}

We can now formulate the dynamic program.
For one leaf, there is only one partial solution $(1, \ldots, 1)$.
For every $i \in \set{2, \ldots, n_0}$, we obtain all partial solutions with $i$ leaves by combining all pairs of partial solutions where the total number of leaves add up to $i$, and keeping only vectors that satisfy the conditions in Lemma~\ref{lemma:dp-correct}.
Lemma~\ref{lemma:dp-correct} states that we only get partial solutions this way, and Lemma~\ref{lemma:dp-complete} proves that any partial solution is generated.
Thus, the correctness of this DP follows directly from the previous two lemmas.

\subsection{Adding Lower Capacity Constraints}
By additionally introducing a lower capacity $\low_i$ for each layer, we need to slightly change the dynamic program.
While a tree that violates the upper capacity constraint cannot be made valid by combining it with another tree, a violation of the lower capacity constraint may be fixed for vertices above the branching layer since the weight of these may increase.
Thus, we are also interested in layer trees where the lower capacity is relaxed for vertices above the branching layer.

Formally, we say that a layer tree $T$ with branching layer $k$ is \emph{almost valid} if it is valid in each layer $i \in [k]$, and if the number of leaves does not exceed the upper capacity of layer $i \in \set{k + 1, \ldots, \layer}$.
A vector is a \emph{relaxed partial solution} if it corresponds to an almost-valid tree.

In a similar fashion as in the case without lower capacity constraints, we show that by combining relaxed partial solutions for fewer leaves, we can obtain all relaxed partial solutions.

\begin{lemma}
  \label{lemma:dp-min-cap-correct}
  Let $a$ and $b$ be relaxed partial solutions with $a_0$ and $b_0$ leaves, respectively.
  The $k$-combination $c = (c_0, \ldots, c_\layer)$ of $a$ and $b$ for some $k$ is a relaxed partial solution with $a_0 + b_0$ leaves if
  \begin{itemize}
    \setlength\itemsep{-1pt}
    \item $c_i$ is at most $n_i$ for every $i \in \setz{\layer}$,
    \item $c_0 \leq \high_{k + 1}$, and
    \item $a_0 \geq \low_k$, and $b_0 \geq \low_k$.
  \end{itemize}
\end{lemma}
\begin{proof}
  Let $a = (a_0, \ldots, a_\layer)$ and $b = (b_0, \ldots, b_\layer)$ be relaxed partial solutions, and let $c$ be the $k$-combination of $a$ and $b$ and properties as stated in the lemma.
  Moreover, let $T_a$ and $T_b$ be almost-valid trees that correspond to $a$ and $b$ with branching layers $k_a \leq k$ and $k_b \leq k$, respectively.

  We prove that $c$ is a partial relaxed solution if the conditions from the lemma are satisfied.
  Let $T_c$ be the $k$-combination of $T_a$ and $T_b$.
  By definition, $T_c$ corresponds to $c$, and each layer $i \in \setz{\layer}$ in $T_c$ consists of at most $n_i$ vertices.
  Further, the weight of the vertex in $T_c$ in the layer directly above the branching layer does not exceed their upper capacities by assumption.
  Since we assume that the instance is normalized, the capacity of a higher layer is at least the capacity of a lower layer.
  Thus, the upper capacity constraint is satisfied for all layers above the branching layers.

  Recall that the weights of vertices in the layers at most $k$ do not change in $T_c$ compared to the corresponding vertices in $T_a$ and $T_b$ by Observation~\ref{observation:equal-load}.
  Thus, it is sufficient to show that each vertex in the layers at most $k$ in $T_a$ and $T_b$ satisfy all capacity constraints.
  As $T_a$ and $T_b$ are almost valid, upper capacity constraints are satisfied in all layers in both trees, and lower capacity constraints are satisfied in the layers at most $k_a$ and $k_b$ respectively.
  In each layer $i \in \set{k_a + 1, \ldots, k}$, $T_a$ only has one vertex to which all leaves are connected.
  As we assume that we have a normalized instance, the lower capacity of a lower layer is at most the lower capacity of a higher layer.
  Since $a_0 \geq \low_k$, the lower capacity constraint is also satisfied in $T_a$ for the layers above its branching layer and at most $k$.
  The symmetrical result holds for $T_b$.

  Thus, $T_c$ is an almost-valid tree, and $c$ a partial relaxed solution.\hfill
\end{proof}

On the other hand, we prove that every relaxed partial solution can be obtained by combining two relaxed partial solutions for fewer leaves.
Observe that the proof of this lemma mostly follows the proof for Lemma~\ref{lemma:dp-complete}.

\begin{lemma}
  \label{lemma:dp-min-cap-complete}
  Let $c \in \mathbb{N}^{\layer + 1}$ be a relaxed partial solution with $c_0 > 2$ leaves with branching layer $k$.
  There exist two relaxed partial solutions $a, b \in \mathbb{N}^{\layer + 1}$ with $a_0 < c_0$ and $b_0 < c_0$ leaves such that
  \begin{itemize}
    \setlength\itemsep{-1pt}
    \item it is $a_0 \geq \low_k$, and $b_0 \geq \low_k$.
    \item $c$ is the $k$-combination of $a$ and $b$.
  \end{itemize}
\end{lemma}
\begin{proof}
  Let $c = (c_0, \ldots, c_\layer)$ be a relaxed partial solution, and let $T_c$ be a corresponding almost-valid layer tree.
  We construct two layer trees $T_a$ and $T_b$ as follows.
  Let $u$ be a vertex in layer $k + 1$ of $T_c$.
  As it is above the branching layer, $u$ has at least two children.
  Let $v$ be a child of $u$ in layer $k$.
  We obtain the layer tree $T_a$ by deleting $v$ and the subtree rooted in $v$ from $T_c$.
  Similarly, $T_b$ is constructed by deleting all children of $u$ except for $v$ and their subtrees in $T_c$.
  By construction, $T_c$ is the $k$-combination of $T_a$ and $T_b$.

  We prove that both $T_a$ and $T_b$ are almost-valid layer trees.
  Each layer in $T_a$ and $T_b$ consists of at most as many vertices as the corresponding layer in $T_c$.
  Since $k$ is the branching layer of $T_c$, vertices in layers at most $k$ in $T_c$ satisfy all capacity constraints.
  The same holds for vertices in layers at most $k$ in $T_a$ and $T_b$ since the weight of such a vertex is equal to the weight of its corresponding vertex in $T_c$.
  Since it is $k \geq k_a$ and $k \geq k_b$, $T_a$ and $T_b$ are almost-valid layer trees.
  Moreover, the lower capacity constraint is also satisfied in layer $k$ in both $T_a$ and $T_b$.\hfill
\end{proof}
If an instance also contains lower capacity constraints, we modify the previously proposed dynamic program as follows.
Instead of only storing partial solutions, we store all relaxed partial solutions.
New relaxed partial solutions for $i \in \set{2, \ldots, n_0}$ leaves are obtained by combining all pairs of relaxed partial solutions where the total number of leaves add up to $i$.
Additionally, we only keep the generated solutions that satisfy the conditions stated in Lemma~\ref{lemma:dp-min-cap-correct}, which tells us that only relaxed partial solutions are generated in this case.
Lemma~\ref{lemma:dp-min-cap-complete} proves that any relaxed partial solution is generated this way.

If the lower capacity of some layer is higher than the total number of leaves to connect, then the instance trivially is a no-instance.
Assuming there is no such layer, each relaxed partial solution with $n_0$ leaves is a partial solution as well.

\subsection{Complexity}
The running time of the algorithm heavily depends on the number of relaxed partial solutions that are stored for each number of leaves.
We denote the number of relaxed partial solutions for $c \in [n_0]$ leaves by $S_c$.
Further, let $S_{\max}$ be the maximum over all $S_c$, i.e., for each number of leaves, we store at most $S_{\max}$ relaxed partial solutions.

To generate the set of relaxed partial solutions with $c \in [n_0]$ leaves, we consider all pairs of relaxed partial solutions with $a$ and $b$ leaves where $a + b = c$.
For each pair of relaxed partial solutions, there are at most $\layer$ possible $k$-combinations.
Thus, we have $$\sum_{\substack{a + b = c \\ a,b \leq c}} \layer \cdot S_a \cdot S_b$$ combinations that need to be considered for $c$ leaves.
These combinations need to be filtered in accordance to Lemma~\ref{lemma:dp-min-cap-correct} which can be naively done in $\mathcal{O}(\layer)$ time for each combination along with the combining.
Further, we use an ordered set for each number of leaves to make sure each relaxed partial solution is only stored once.

This gives us the running time
\begin{align*}
  &\mathcal{O}\biggl(\sum_{c = 2}^{n_0} (\log(S_c) + \layer) \cdot  \sum_{\substack{a + b = c \\ a,b \leq c}} \layer \cdot S_a \cdot S_b \biggr) \\
  &\subseteq \mathcal{O}\biggl(\sum_{c = 2}^{n_0} (\layer \cdot \log(S_{\max}) + \layer^2) \cdot  \sum_{\substack{a + b = c \\ a,b \leq c}} S_{\max}^2 \biggr) \\
  & \subseteq \mathcal{O}\biggl(S_{\max}^2 \cdot (\layer \cdot \log(S_{\max}) + \layer^2) \cdot \sum_{c = 2}^{n_0} c \biggr) \\
  & \subseteq \mathcal{O}\biggl(n_0^2 \cdot S_{\max}^2 \cdot (\layer \cdot \log(S_{\max}) + \layer^2) \biggr) \, .
\end{align*}

Since at most $n_0$ vertices are needed in each layer, we get a trivial bound of $S_{\max} \leq n_0^\layer$, which yields the following theorem.

\begin{theorem}
  The \sofaclapfeas problem with $n_0$ leaves and $\lambda$ layers can be solved in $\mathcal{O}\bigl(n_0^{2 \layer + 2} \cdot \layer^2 \cdot \log n_0 \bigr)$ time.
\end{theorem}

If we assume $\layer$ to be a fixed constant, then the running time of the proposed algorithm is polynomial in $n_0$, the number of sources to connect.
However, as already mentioned in the introduction, instances of \sofaclapfeas can be compactly represented by encoding the numbers in binary.
Thus, the DP is a pseudo-polynomial algorithm for \sofaclapfeas, assuming $\layer$ is constant.
This partially confirms the conjecture by Gritzbach et al.\ \cite[Conjecture 1]{gritzbach2022sofaclap}.

\section{Optimizations for the DP Algorithm}
\label{sec:optim-dp-algor}

We propose two types of optimizations for our DP.
First, we show that some partial solutions can be ignored without violating the correctness of the algorithm.
Secondly, one can sometimes stop the DP early by completing one of the partial solutions to a full solution.

\subsection{Reducing Stored Solutions}

We provide three ways to reduce the set of partial solutions that needs to be computed.
First, we show that it suffices to consider Pareto-optimal solutions, i.e., it is always better to connect the same number of leaves with fewer vertices in higher levels.
Secondly, we provide a way of computing upper bounds on how many sources can be at most connected when using a given partial solution.
If this upper bound is below the number of sources we need to connect, we can prune the partial solution.
Thirdly, we prove that any partial solution can be assumed to be balanced in the sense that it is composed of two smaller partial solutions with roughly the same number of leaves.

\subsubsection{Pareto-Optimal Solutions}

Consider two different relaxed partial solutions $a = (a_0, \ldots, a_\layer)$ and $b = (b_0, b_1, \ldots, b_\layer)$ with the same number of leaves $a_0 = b_0$.
We say that $a$ \emph{dominates} $b$ if $a_i \le b_i$ for all $i \in [\layer]$.
We say that a relaxed partial solution is \emph{Pareto optimal} if it is not dominated by any relaxed partial solution.

Recall that for the correctness of our dynamic program, we showed in Lemma~\ref{lemma:dp-min-cap-complete} that we can obtain any relaxed partial solution as the $k$-combination of two partial solutions with fewer leaves.
The next lemma strengthens this result by stating the same result for Pareto-optimal partial solutions.
This directly implies that restricting the dynamic program to only include Pareto-optimal partial solutions is correct.

\begin{lemma}
  \label{lemma:Pareto-dp-correct}
  Let $c \in \mathbb{N}^{\layer + 1}$ be a Pareto-optimal relaxed partial solution for $c_0 > 2$ leaves with branching layer $k$.
  There exist two Pareto-optimal relaxed partial solutions $a, b \in \mathbb{N}^{\layer + 1}$ for $a_0 < c_0$ and $b_0 < c_0$ leaves such that
  \begin{itemize}
    \item it is $a_0 \geq \low_k$ and $b_0 \geq \low_k$, and
    \item $c$ is the $k$-combination of $a$ and $b$.
  \end{itemize}
\end{lemma}

\begin{proof}
  Let $c$ be a Pareto-optimal relaxed partial solution for $c_0$ leaves with branching layer $k$.
  By Lemma~\ref{lemma:dp-min-cap-complete}, there are two relaxed partial solutions $a$ and $b$ that satisfy all conditions above apart from being Pareto optimal.
  We show that $a$ and $b$ are indeed Pareto optimal.
  Assume that $a$ is not Pareto optimal, and that it is dominated by some relaxed partial solution $a'$ with the same number of leaves.
  Observe that $a'$ and $b$ together satisfy all conditions from Lemma~\ref{lemma:dp-min-cap-correct}.
  This directly implies that the $k$-combination of $a'$ and $b$ is a relaxed partial solution with $c_0$ leaves as well.
  Moreover, the $k$-combination of $a'$ and $b$ dominates $c$ since $a'$ has at most as many vertices as $a$ in each layer.
  Thus, $c$ is not Pareto optimal which contradicts the assumption.
  Using the symmetrical argument for $b$, we showed that both $a$ and $b$ are Pareto-optimal relaxed partial solutions which concludes the proof.\hfill
\end{proof}

\subsubsection{Pruning}
\label{sec:pruning}

The following lemma provides an upper bound on the number of leaves in a valid layer tree that depends on the number of vertices and the upper bounds given for a pair of layers.

\begin{lemma}
  \label{lemma:two-layers}
  Let $0 < i < j$ be two layers in an instance of \sofaclapfeas with $n_i$ and $n_j$ vertices and upper capacities $\high_i$ and $\high_j$, respectively.
  Then the maximum number of leaves in a valid layer tree is upper bounded by $n_i \high_i$, $n_j \high_j$, and
  \begin{equation*}
    \biggl\lfloor \frac{\high_j}{\high_i} \biggr\rfloor \high_i n_j + (\high_j \bmod \high_i) \cdot \max\biggl\{0,n_i - n_j \cdot \biggl\lfloor \frac{\high_j}{\high_i}\biggr\rfloor\biggr\}.
  \end{equation*}
\end{lemma}
\begin{proof}
  The upper bounds of $n_i \high_i$ and $n_j \high_j$ are obvious.
  For the remaining upper bound, we optimistically assume that we can choose the weight of each vertex in layer $i$ arbitrarily (i.e., the number of leaves below this vertex) and that we can arbitrarily choose which vertices in layer $i$ are descendants of which vertices in layer $j$, i.e., how the weights in layer $i$ are grouped to form the weights in layer $j$.

  Consider one vertex $x$ in layer $j$ that has descendants $v_1, \dots, v_k$ in layer $i$.
  Then, when maximizing the number of leaves, we can assume without loss of generality that $w(v_i) = \high_i$ except for at most one $i \in [k]$.
  This is true as we can otherwise shift weight between the $v_i$, which can only free up vertices on layer $i$.
  Now we count the total weight in layer $j$ by separately counting descendants in layer $i$ that have weight $\high_i$ and those that have lower weight.

  Clearly, vertex $x$ in layer $j$ can have at most $\lfloor \frac{\high_j}{\high_i} \rfloor$ descendants with the full weight of $\high_i$, which together contribute a weight of at most $\lfloor \frac{\high_j}{\high_i} \rfloor \cdot \high_i$ per vertex.
  As there are at most $n_j$ vertices in layer $j$, this gives the first term in the claimed bound.

  Note that if there are also descendants with weight less than $\high_i$, then every vertex in layer $j$ actually has $\lfloor\frac{\high_j}{\high_i}\rfloor$ descendants with full weight.
  Thus, there are at most $\max\{0, n_i - n_j\cdot \lfloor\frac{\high_j}{\high_i}\rfloor\}$ descendants that have weight less than $\high_i$.
  Each of them contributes at most $(u_j \mod u_i)$ to the weight its corresponding vertex in layer $j$, which concludes the proof. \hfill
\end{proof}
We use this lemma as follows.
Let $a = (a_0, \dots, a_\layer)$ be a relaxed partial solution.
One can observe that the problem of extending $a$ to a full solution is again an instance of \sofaclapfeas, where we reduce the number of vertices $n_i$ available on each layer by $a_i$ if $a_i > 1$.
We then test whether the upper bound from Lemma~\ref{lemma:two-layers} is below $n_0 - a_0$ in which case the partial solution $a$ can be discarded as it cannot be extended to a full solution with $n_0$ leaves.
We compute this upper bound for every pair of consecutive layers, i.e., for layers $i$ and $j = i + 1$ for $i \in [\layer - 1]$.

We note that this upper bound does not use the lower capacities.
We get an additional upper bound if $\lfloor \frac{\high_j}{\high_i} \rfloor\cdot \high_i < \low_j$.
In this case, the lower bound $\low_j$ on layer $j$ is so high that only having descendants on layer $i$ with weight $\high_i$ is not sufficient to reach $\low_j$.
In this case, we additionally get the upper bound $\high_j \cdot \lfloor n_i / \lceil \frac{\high_j}{\high_i} \rceil \rfloor$.

\subsubsection{Balanced Solutions}

To obtain all relaxed partial solutions with $c_0$ sources, the dynamic program combines all pairs of relaxed partial solutions with $a_0$ and $b_0$ sources such that $c_0 = a_0 + b_0$.
Here we show that it is sufficient to consider only combinations that are \emph{balanced} in the sense that $a_0, b_0 \ge \frac{c_0}{3}$.
\begin{lemma}
  \label{lemma:balanced-solution}
  For every Pareto-optimal relaxed partial solution $c$ with $c_0 > 1$ leaves with branching layer $k$, there are two Pareto-optimal relaxed partial solutions $a$ and $b$ with $a_0$ and $b_0$ leaves, respectively, such that
  \begin{itemize}
    \item $a_0 \geq \frac{c_0}{3}$ and $b_0 \geq \frac{c_0}{3}$,
    \item $c$ is the $k$-combination of $a$ and $b$.
  \end{itemize}
\end{lemma}
\begin{proof}
  We prove the lemma by induction on the number of sources $c_0$ to connect.
  For $c_0 = 2$, there is at most one relaxed partial solution $(2, 1, \ldots, 1)$, which can be partitioned into two relaxed partial solutions $(1, 1, \ldots, 1)$.
  Such a partition is obviously balanced.

  Let $c_0 > 2$ and $c$ be a Pareto-optimal relaxed partial solution for $c_0$ sources with branching layer $k$.
  By Lemma~\ref{lemma:Pareto-dp-correct}, there are two Pareto-optimal relaxed partial solutions $a$ and $b$ for $a_0$ and $b_0$ leaves, respectively, and $a_0, b_0 \geq \low_k$ such that $c$ is the $k$-combination of $a$ and $b$.
  If both $a$ and $b$ connect at least $\frac{c_0}{3}$ sources, then we are done.
  Thus, we assume without loss of generality that $a_0 < \frac{c_0}{3}$ and $b_0 > \frac{2c_0}{3}$.

  The main idea is to rebalance $a$ and $b$ by removing a part of $b$ and combining it with $a$.
  By induction, $b$ has a balanced partition into two Pareto-optimal relaxed partial solutions $d$ and $e$ with $d_0$ and $e_0$ leaves, respectively, such that $d_0, e_0 \geq \frac{b_0}{3}$ and $d_0, e_0 \geq \low_{k_b}$, where $k_b$ is the branching layer of $b$.
  Without loss of generality, we assume that $d_0 \leq e_0$.
  Let $ad$ be the $k_b$-combination of $a$ and $d$ with $ad_0$ leaves.
  We argue that the $k$-combination of $ad$ and $e$ equals $c$ and is balanced.
  Because of $k_a, k_d \leq k_b \leq k$ and $k_e \leq k$ (all inequalities follow from Pareto-optimality), both combinations are well defined.
  Figure~\ref{fig:balanced-tree} depicts a layer tree corresponding to the described setting.

  \begin{figure}
    \centering
    \includegraphics[page=1,width=\columnwidth]{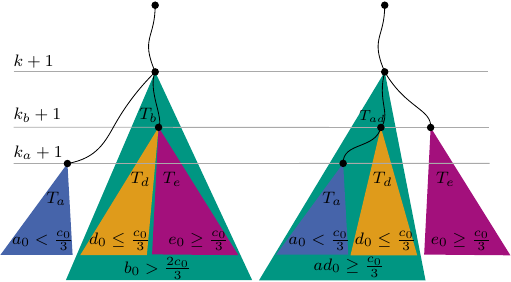}
    \caption{The setting before and after rebalancing. We $k_b$-combine $T_a$ and $T_d$ first, before $k$-combining the result with $T_e$.}
    \label{fig:balanced-tree}
  \end{figure}

  We first prove that this partition is indeed balanced, i.e., $e_0 \geq \frac{c_0}{3}$ and $ad_0 \geq \frac{c_0}{3}$.
  By assumption, it is $d_0 \leq e_0$ which directly implies $e_0 \geq \frac{b_0}{2} > \frac{c_0}{3}$.
  Further, we have
  $ad_0 = a_0 + d_0 \geq a_0 + \frac{b_0}{3} = a_0 + \frac{c_0 - a_0}{3} = \frac{c_0}{3} + \frac{2 a_0}{3} \geq \frac{c_0}{3}$
  for the number of leaves of $ad_0$.

  Further, we show that both $ad$ and $e$ relaxed partial solutions.
  For $e$, this is true by definition.
  As $ad$ is the combination of two relaxed partial solutions, it is sufficient to show that the conditions from Lemma~\ref{lemma:dp-min-cap-correct} are satsified.
  First, observe that $ad$ consists of at most as many vertices in each layer as $c$.
  For the upper capacity constraint, we show that $\high_{k_b + 1} \geq ad_0 = a_0 + d_0$.
  Since $b$ is a relaxed partial solution with branching layer $k_b$, it is $\high_{k_b + 1} \geq b_0 = d_0 + e_0 \geq d_0 + \frac{c_0}{3} \geq d_0 + a_0$.
  Further, note that $a_0 \geq \low_k \geq \low_{k_b}$ and $d_0 \geq \low_{k_b}$ hold by construction.
  By Lemma~\ref{lemma:dp-min-cap-correct}, $ad$ is a relaxed partial solution since $a$ and $d$ are relaxed partial solutions.

  Now, we consider the $k$-combination $c'$ of $ad$ and $e$ with $a_0 + d_0 + e_0$ leaves.
  It remains to show that $c'$ is equal to $c$.
  For the layers above $k$, both $c'$ and $c$ only have one vertex by definition.
  In the layers $i \in \set{k_b + 1, \ldots, k}$, $c$ has two vertices because $k_a \leq k_b$.
  The same holds for $c'$ since $ad$ consists of one vertex above its branching layer $k_b$, as does $e$ whose branching layer is at most $k_b$.
  For the layers $i \in \setz{k_b}$, both $c$ and $c'$ have the sum of the number of vertices in $a$, $d$ and $e$ in this layer.

  To conclude, the relaxed partial solutions $ad$ and $e$ are a balanced partition of $c$. \hfill
\end{proof}
This lemma has two main algorithmic consequences.
First, we never need to combine partial solutions with $a_0$ and $b_0$ leaves if $a_0 \le \frac{a_0 + b_0}{3}$ or $b_0 \le \frac{a_0 + b_0}{3}$.
Moreover, the full solution with $n_0$ leaves can be obtained by combining two partial solutions with at least $\frac{n_0}{3}$ leaves and thus we do not need to compute partial solutions with more than $\frac{2n_0}{3}$ leaves.

\subsection{Completing a Partial Solution}

Here we present two different ways in which we try to stop the dynamic program prematurely by finding a full solution before the dynamic program actually reaches $n_0$ leaves.

\subsubsection{Checking Counterparts}

For a relaxed partial solution $a$ with $a_0 > \frac{n_0}{2}$ leaves, we have already generated all relevant counterparts with $b_0 = n_0 - a_0 < \frac{n_0}{2}$ leaves that could be composed with $a$ to form a full solution.
Thus, we can check for each of these relaxed partial solutions if the conditions from Lemma~\ref{lemma:dp-min-cap-correct} are fulfilled.
If yes, then we can stop and output the solution.
If there is no such counterpart for $a$, then $a$ is not part of a full solution, and we can discard $a$.
In any case, this means that we do not need to store any relaxed partial solutions for more than $\frac{n_0}{2}$ leaves.

\subsubsection{Greedy Completion}
\label{sec:greedy-completion}

If we have a relaxed partial solution $a$ with $a_0 < \frac{n_0}{2}$ leaves, then we do not know any possible counterparts yet, other than in the previous case.
However, as already mentioned in Section~\ref{sec:pruning} on pruning, the problem of extending $a$ to a full solution is again an instance of \sofaclapfeas.
We use a greedy heuristic to search for a solution of this remaining instance.

The idea here is to construct a layer tree using on the remaining vertices connecting $b_0 = n_0 - a_0$ leaves.
Our greedy approach is aimed at constructing almost valid solutions.
We ignore the lower capacity constraints for now and assume that $\ell_i=0$ for all $i \in \set{0,\dots,\lambda}$.
For better intuition of our greedy approach, we introduce a different view on the problem: The leaves are units of flow and the vertices can hold a certain amount of incoming flow -- their capacity.
Further flow that cannot be pushed to a higher layer is wasted.
The goal is to push as much flow from layer one to the highest layer which is equivalent to connecting as many leaves as possible.
Starting with layer 1 where we assume every vertex is completely saturated, we push flow upwards as follows.
Assume that we have already pushed the flow to layer $i - 1$ with $i \in \set{2,\dots,\lambda}$.
Pushing flow from layer $i - 1$ to layer $i$ is a bin-packing problem where each vertex in layer $i$ has capacity $u_i$ and the weight of each vertex in layer $i - 1$ is its incoming flow.
We sort the vertices in layer $i-1$ by their incoming flow in increasing order and assign each vertex greedily to a vertex in layer $i$ with the currently lowest incoming flow.
This is repeated for each layer and the incoming flow at the root equals the maximum number of leaves we can connect using this greedy approach.

Observe that without lower capacity constraints, it is also possible to connect any number of leaves lower than that since in this case any layer sub-tree is also valid.
Thus, without lower capacity constraints, we know that a vector $b$ is a partial solution if the incoming flow at the root of a greedily constructed layer tree using the vertices in $b$ is at least $b_0$.
In particular, we can easily construct a valid layer tree with exactly $b_0$ leaves in this case.

However, this greedy construction takes linear time and can be very time consuming if done for every relaxed partial solution.
Determining the incoming flow at the root can be sped up by not constructing the entire tree explicitly.
Instead, we only keep track of how many vertices in each layer hold a given amount of flow and process vertices with the same amount of flow together as one batch.
In particular, it can be shown that using our greedy approach, there are at most $i$ vertices with pairwise distinct amounts of incoming flow in layer $i\in[\lambda]$.
Using this optimization, the amount of flow that reaches the highest layer can be computed in $\mathcal{O}(\lambda^2 \cdot \log(\lambda))$.

If we have additional lower constraints, then not every layer sub-tree of a greedily constructed layer tree $T$ is necessarily almost valid itself.
Thus, in this case, we want to find an almost-valid layer sub-tree of $T$ connecting exactly $b_0$ leaves.
Note that we do not change the topology of the greedily constructed tree after its construction.

The idea here is to distribute exactly $b_0$ units of flow from the highest to the lowest layer in $T$ such that no flow is wasted and all capacity constraints are satisfied.
How much flow can be distributed within a sub-tree rooted at some vertex depends on the topology of its sub-tree.
The minimum and maximum number of flow units that can be distributed in the sub-tree of a vertex $v$ is captured in $f_{\min}(v)$ and $f_{\max}(v)$, respectively.
Each vertex needs to satisfy its own capacity constraints on the minimum and maximum flow units of its children.
For a vertex $v$ in layer $i > 1$, $f_{\min}(v)$ is the maximum of its lower capacity and the sum of $f_{\min}$ over its children.
Analogously, $f_{\max}(v)$ is the minimum of its upper capacity and the sum of $f_{\max}$ over its children.

Observe that if a vertex $v$ is assigned flow in $\set{0} \cup [f_{\min}(v), f_{\max}(v)]$, then it is possible to distribute the flow within the sub-tree of $v$ such that each vertex in the sub-tree satisfies all constraints.
Further note that if some vertex $v$ is assigned more than $f_{\max}(v)$, then $v$ immediately violates the upper capacity constraint.
However, $v$ max be assigned less than $f_{\min}(v)$ if $f_{\min}(v)$ is more than its lower capacity and if some vertex in the sub-tree is not used at all.
We distribute the units of flow greedily from top to bottom such that all vertices are assigned at most $f_{\max}(v)$ flow and in each layer the number of vertices with at least $f_{\min}(v)$ flow is at high as possible.
The flow units assigned to a vertex in layer 1 correspond to the number of leaves that we connect to it.
If the capacity constraints of all vertices are satisfied, we obtain an almost-valid layer tree with the remaining vertices and exactly $b_0$ leaves.
This takes linear time in the size of the constructed tree.

\subsubsection{Estimating the Greedy Heuristic}
The running time of this greedy heuristic is linear in the number of vertices, which tends to be too expensive to be run for every partial solution.
However, we have two ways to predict whether running the greedy heuristic is promising.
A first optimization is to only greedily construct a tree if we reach the required number of leaves while ignoring the lower capacities.
This can be done in $\mathcal{O}(\lambda^2 \cdot \log(\lambda))$, as stated above.
If this does not yield a tree connecting all sources, then the greedy heuristic that respects the lower bound cannot find a valid solution and we can skip it.
This approach will never discard a feasible completion.

Secondly, we observe that consecutive runs of the greedy heuristic are often on somewhat similar instances, i.e. they contain a similar number of sources.
Thus we remember the result of the previous run (in terms of how many vertices were used on which layer) and then use this as the basis for a heuristic to decide whether the greedy algorithm should be run.
More specifically, each time we greedily construct a layer tree $T$ from some vector $b$, we store both $a$ and the number of leaves connected by $T$.
Using this information and the capacity constraints of each layer, we estimate for further vectors how many leaves they connect using the greedy heuristic.
If this number is lower than the number of leaves to connect, we do not try to construct a corresponding tree.
This estimation is determined in $\mathcal{O}(\layer)$.
In contrast to the first predictor, this estimate may prune vectors for which the greedy construction would succeed.
This does, however, not affect the validity of the dynamic program as the greedy completion is only a shortcut and the partial solution can still be competed by the dynamic program.

\section{Heuristic Optimization}
\label{sec:embedding}

As mentioned in the introduction, studying the abstract problem \sofaclapfeas is motivated by the fact that it appears to be the hard core of \sofaclap (and potentially other related network design problems).
To justify this motivation, we evaluate how solving \sofaclapfeas with our dynamic program can bootstrap the MILP formulation of \sofaclap to find a good solution in Section~\ref{sec:evaluation}.
In this section, we additionally provide heuristics that aim to find a good solution for \sofaclap, given a solution for the corresponding \sofaclapfeas problem.
We evaluate them in comparison to the MILP in Section~\ref{sec:evaluation}.

Recall that for \sofaclap, we are given a layer graph with edge lengths and a solution consists of a layer tree that is a subgraph of the layer graph.
Being a subgraph can be expressed with an injective function that maps each vertex of the layer tree to a vertex in the corresponding layer of the layer graph.
We call this function an \emph{embedding} and we call the vertices of the layer graph \emph{positions} to which the vertices in the layer tree are mapped.
This makes it easier to distinguish between the vertices in these different graphs.
We prove that this embedding problem alone is already NP-hard in Section~\ref{sec:optim-embedd-hard}.
Instead of computing an optimal embedding, we thus propose several simple heuristics in Section~\ref{sec:embedding heuristics}

\subsection{Optimal Embedding is Hard}
\label{sec:optim-embedd-hard}
Our dynamic program produces the tree structure of a feasible solution for a fully connected \sofaclap instance.
It does, however, not take into account the given edge weights.
We show that finding an optimum mapping from a $\lambda$-forest to the vertices of the \sofaclap layer graph is strongly NP-hard.

Given a \sofaclap instance $S=(G=(V, E_G), l_1, u_1, \dots, l_\layer,u_\layer, \cost, \len)$ and a valid layer tree $F=(F_0 \dotcup \dots \dotcup F_\layer, E)$ for that instance, we want to find a set of mappings $\tau_i: F_i \to V_i$ for each layer $i \in [\layer]$ to obtain a cable layout with minimum weight.
A mapping is an injective function, i.e., vertices in $F$ are mapped to distinct vertices.
In this setting, $s(u)$ and thus $\cost(s(u))$ is a fixed number for each edge $(u,v) \in E$.

The corresponding decision problem is defined as follows:

\begin{Definition}[\sofaclapembed]
  Given a solar farm instance $S$, a layer tree $F$ and a number $k$, is there a set of mappings $\tau_i$ for each $i \in [\layer]$ with total weight \[\sum_{i=1}^{\layer}\sum_{(u, v)\in E_i(F)} \cost(s(u)) \cdot \len (\tau(u), \tau(v)) \leq k \; ?\]
\end{Definition}

This problem is strongly NP-hard, even if vertices are assigned positions in the Euclidean plane and the length of an edge is the Euclidean distance between the two incident vertices.
We show this by using a reduction from \threepartition, a strongly NP-hard problem that is defined as follows.
Let $m, T \in \Natural$ and let $A = \set{a_1, \ldots, a_{3m}}$ be a set of $3m$ numbers with $\frac{T}{4} < a < \frac{T}{2}$ for all $a \in A$ and $\sum_{a \in A} a = mT$.
Is there a partition of $A$ into $m$ sets $S_1, \ldots, S_m$ such that $|S_i| = 3$ and $\sum_{a \in S_i} a = T$ for each $i \in [m]$?

\begin{lemma}
  \sofaclapembed is strongly NP-hard, even when $\len$ is the Euclidean distance between the edge endpoints and $\cost(k) = 1$ for all $k \in \Natural$.
\end{lemma}

\begin{proof}
  We use a reduction from \threepartition.
  Given a \threepartition-instance $I = (m, X, A)$, we construct an instance $I' = (S, F, k)$ of \sofaclapembed.
  The solar farm $S$ with two layers $V_0$ and $V_1$, which represents the triplets $S_i$, is constructed as follows.
  Layer $1$ consists of $3m$ vertices $\bigcup_{1 \leq i \leq m} W_i$, where $W_i = \set{w_{i, 1}, w_{i, 2}, w_{i, 3}}$, with upper capacity $\frac{T}{2}$ and lower capacity $\frac{T}{4}$.
  Layer $0$ consists of $mT$ vertices $\bigcup_{1 \leq i \leq m} U_i$, where $U_i = \set{u_{i, 1}, \ldots, u_{i, X}}$.
  Further, we implicitly define the length of each edge by assigning a position $\mathrm{pos}(u)$ in the Euclidean plane to each vertex $u \in V_0 \cup V_1$.
  Vertices in layer $0$ are distributed along the line $y = 0$, starting with $x = 1$.
  Vertices of the same group are assigned consecutive, distinct $x$-coordinates, and we leave a space of $m \cdot T \cdot (T + 1)$ between the last vertex and the first vertex of consecutive groups.
  The vertices in layer $1$ are distributed along the line $y = 1$, and the first vertex of each group in $V_1$ aligns with the first vertex of each group in $V_0$.
  More formally, it is
  \begin{align*}
    \mathrm{pos}(u_{i, j}) &= ((X + m \cdot T \cdot (T + 1)) \cdot i + j, 0) \text{ and } \\
    \mathrm{pos}(w_{i, j}) &= ((X + m \cdot T \cdot (T + 1)) \cdot i + j, 1) \; .
  \end{align*}

  Further, we construct a forest $F = ((F_0, F_1), E)$ that is to be embedded in $I$.
  As for the solar farm $S$, the layer $F_0$ consists of $mX$ vertices and the layer $F_1$ consists of $3m$ vertices, we denote by $f_1, \ldots, f_{3m}$.
  Each vertex $f_i \in F_1$ represents a distinct number in the set of the \threepartition-problem, and is connected to exactly $a_i$ vertices in $F_0$.
  Note that since each $a_i$ is at most $X$, $F$ is a feasible solution for the given instance.

  Finally, let $k \coloneqq m \cdot T \cdot (T + 1)$.
  An example of such an instance $I'$ with am embedding of length at most $k$ is seen in Figure~\ref{fig:embedding-reduction-example}.

  \begin{figure}
    \centering
    \includegraphics[width=\columnwidth]{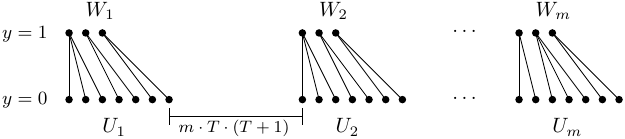}
    \caption{A possible solar farm $S$ of an instance $I'$ of \sofaclapembed with an embedding of some forest $F$.
      Since the embedding has total length of at most $k = mT(T+1)$, the depicted instance is a YES-instance of \sofaclapembed.}
    \label{fig:embedding-reduction-example}
  \end{figure}

  We now show that $I'$ is a yes-instance, i.e., there is a mapping $\tau$ between the vertices of $S$ and $F$, if and only if $I$ is a yes-instance.

  If $I$ is a yes-instance, $A$ can be partitioned into $m$ triplets $S_i$ with $\sum_{a \in S_i} a = F$.
  For each $a_j \in S_i$, we map the vertex $f_j \in F_1$ to a distinct vertex in $V_1$.
  The children of $f_j$ are mapped to distinct vertices in $U_i$.
  Since each $S_i$ consists of exactly three elements and since $\sum_{a_j \in S_i} s(f_j) =  \sum_{a_j \in S_i} a_j = F$, such a mapping exists.
  Note that vertices in $U_i$ are now only connected to vertices in $W_i$.
  The distance between any vertex in $U_i$ and any vertex in $W_i$ is less than $F$ by construction, thus the weight of the edges between $U_i$ and $W_i$ is less than $T \cdot (T + 1)$.
  This yields a total weight of at most $m \cdot T \cdot (T + 1)$, and it follows that the defined embedding is a solution.

  On the other hand, if $I'$ is a yes-instance, then there is an embedding of $F$ such that the weight does not exceed $m \cdot T \cdot (T + 1)$.
  Let $S_i = \set{a_j \mid \tau(f_j) \in W_i}$.
  Note that the upper capacity and lower capacity constraints of the vertices in $V_1$ enforce that each vertex in $V_1$ is used, i.e., $|S_i| = 3$ for each $i$.
  Further, we observe that the distance between any vertex in $U_i$ and any vertex in $W_j$ with $i \neq j$ is greater than $k$ by construction.
  Thus, vertices in $U_i$ are only connected to vertices in $W_i$, and vertices in $W_i$ are only connected to vertices in $U_i$.
  As a consequence, $\sum_{a_j \in S_i} a_j = \sum_{\tau(f_j) \in W_i} s(\tau(f_j)) = T$ holds, and the $S_i$ form a solution for $I$.
\end{proof}

\subsection{Heuristic Embedding}
\label{sec:embedding heuristics}

While computing a globally optimal embedding is NP-hard, we can make local improvements to find a good embedding.
Our heuristics for \sofaclap start with a valid layer tree by our DP together with an arbitrary embedding into the layer graph.
We then iteratively try to improve the solution.
Our first heuristic does not change the structure of the tree itself but just improves the embedding.
Our other two heuristics also make improvements by changing the structure of the graph.
We apply the heuristics one after the other exhaustively in the order they appear here and iterate until no more changes occur.

\paragraph{Layer-Wise Embedding}

Optimizing the embedding for just one layer while keeping all other layers fixed boils down to a bipartite matching problem.
Each vertex of the layer tree, has to be matched to one position and the choices for different vertices have no influence on each other, despite occupying their assigned position.
We solve this matching problem for each layer individually (without actually changing the embedding) and then adjust the layer that yields the biggest reduction in cost.
This is repeated until no further improvement is possible.
Note that a minimum cost bipartite matching can be computed efficiently, for example using the Hungarian algorithm~\cite{Some_Techn_Transp_Probl-Tomizawa71}.

\paragraph{Equal-Weight Edge Swaps}

Let $v_1, \dots, v_k$ be the vertices of the layer tree in one fixed layer that have the same weight, i.e., the vertices $v_i$ for $i \in [k]$ are all in the same layer and they each have the same number of leaves in their subtree.
Moreover, for every $i \in [k]$, let $p_i$ be the parent of $v_i$.
Then any matching between the children $v_j$ and their parents $p_j$, where parents are allowed to appear multiple times, yields a valid layer tree.
Moreover, when fixing the embedding, the cost of each matching edge is clear and independent of the other choices in this matching.
Thus, an optimal rewiring (conditioned on no additional changes, e.g., to the embedding) again boils down to computing a bipartite matching.
We do this once for all layers, starting at the leaves.

\paragraph{General Edge Swaps}

Here, we allow two vertices $v_1$ and $v_2$ on the same layer with parents $p_1$ and $p_2$, respectively, to swap their parents if this keeps the tree valid.
This can be checked by walking up the path from $p_1$ and $p_2$ to their lowest common ancestor.
We note that in this case, we cannot use a matching algorithm as the validity of a swap may depend on other swaps.
Our heuristic goes through the layers bottom up.
In each layer, we iteratively perform the swap that yields the biggest improvement, until no more improvements are possible.

\section{Evaluation}
\label{sec:evaluation}

We start by evaluating the performance of our dynamic program in comparison to Gurobi for \sofaclapfeas in Section~\ref{sec:eval-dp-gurobi}.
In Section~\ref{sec:eval-opts} we evaluate the impact of the optimizations described in Section~\ref{sec:optim-dp-algor} on the running time of our dynamic program.
Afterwards, in Section~\ref{sec:eval-scale}, we study how the running time scales with growing input size.
Although there is some dependence on the input size, we identify other factors that more strongly determine the difficulty of instances than their sheer size in Section~\ref{sec:diff-instances}.
Finally, in Section~\ref{sec:eval-sofaclap}, we consider the impact of our dynamic program on solving the \sofaclap problem and evaluate our heuristics.

\subsection{Setup}
\paragraph{Test Data}
We use the data set provided in \cite{gritzbach2022sofaclap}.
These instances have been generated to match real-world solar farms and are grouped into three size categories and come with positions.
We use only the largest instances with between 1200 and 1500 strings.

\label{sec:random-test-data}
We generate additional randomized test data parameterized by the number of layers $\layer$, the allowed interval for the number of sources $[n_\ell, n_h]$ and a scaling factor for adjacent layers in the interval $[f_\ell, f_h]$.
The latter influences the number of vertices in the higher layers.
We first draw the number of sources $n_0$ from the from $\set{n_\ell,\dots,n_h}$.
For each higher layer $i>0$ we draw the layer size $n_i$ from $\set{\floor{\frac{n_{i-1}}{f_h}},\dots,\floor{\frac{n_{i-1}}{f_\ell}}}$
The upper capacity $u_i$ of each layer $i > 0$ is chosen uniformly from $\set{\ceil{\frac{n_0}{n_i}},\dots,2 \ceil{\frac{n_0}{n_i}}}$.
The lower capacity $l_i$ is chosen uniformly at random between $0$ and $\frac{2}{3} u_i$ and it limited to the number of sources.

\paragraph{Environment}
Our experiments were run on a server with an Intel\textsuperscript{\textregistered} Xeon\textsuperscript{\textregistered} Gold 6144 8 core CPU clocked at 3.5GHz and 192GB DDR4 memory running Ubuntu 22.04.4 LTS.
The implementation is written in C++ 23 and were compiled with gcc 14.1.0.
We solve the MILP formulations using Gurobi 11.0.2 restricted to a single thread.
If not stated otherwise, we use a time limit of one hour.
Both the source code of our implementation (\href{https://doi.org/10.5281/zenodo.13868979}{DOI 10.5281/zenodo.13868979}) and the evaluated data set (\href{https://doi.org/10.5281/zenodo.13807817}{DOI 10.5281/zenodo.13807817}.) are available on Zenodo.

\subsection{Solving \sofaclapfeas}
\label{sec:eval-dp-gurobi}
\begin{figure}
  \centering
  \includegraphics[width=\columnwidth]{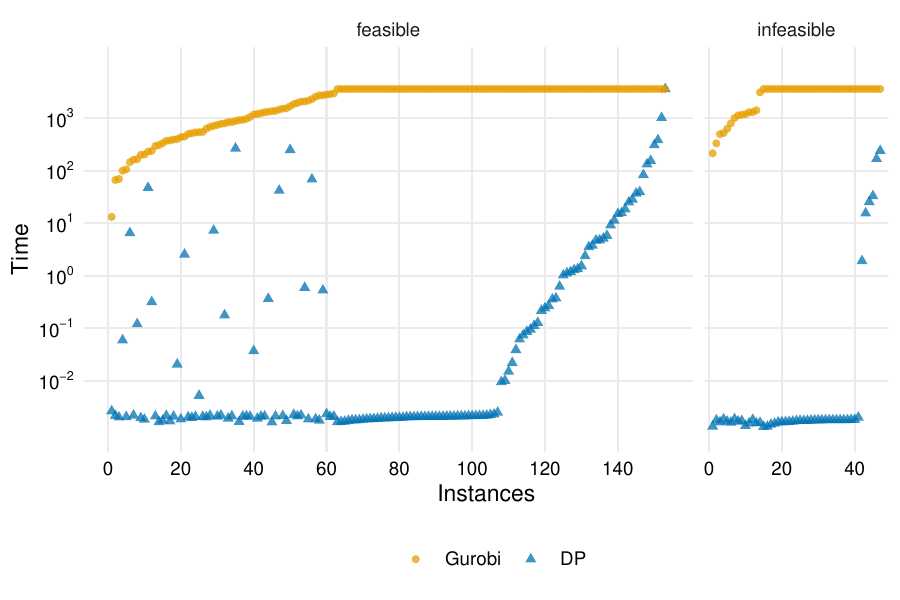}
  \caption{A comparison of the running times between Gurobi and our dynamic program.
  The instances are separated by feasibility and sorted by the running time of Gurobi.}
  \label{fig:eval-time-opts}
  \vspace{-1em}
\end{figure}
\begin{figure*}
  \centering
  \begin{subfigure}[t]{1.23\columnwidth}
    \centering
  \includegraphics[width=0.95\columnwidth]{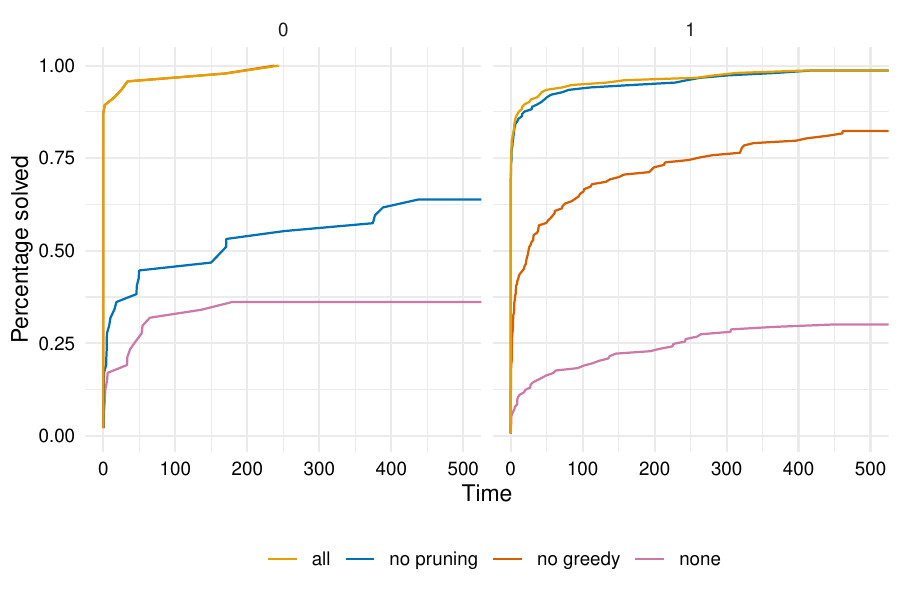}
  \end{subfigure}
  \hfill
  \begin{subfigure}[t]{0.77\columnwidth}
    \centering
    \includegraphics[width=0.95\columnwidth]{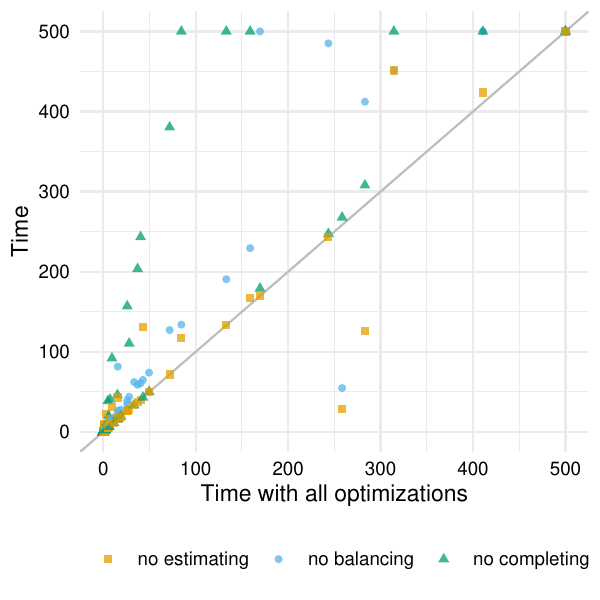}
  \end{subfigure}
  \caption{A comparison of the running times with selected optimization sets.
  The left plot shows the percentage of instances solved within a certain amount of time.
  The right plot compares the running times of some optimization set to the running times of the DP with all optimizations}
  \vspace{-1em}
  \label{fig:eval-opts}
\end{figure*}
We compare the running time of our DP and the MILP for \sofaclapfeas on 200 instances with exactly six layers, 1000 sources and scaling factors between 1.7 and 2.3, generated as in Section~\ref{sec:random-test-data}.
We use the basic dynamic program with all optimizations as described in Section~\ref{sec:optim-dp-algor}.
We see in Figure~\ref{fig:eval-time-opts} that Gurobi only solved less than half the instances within the time limit.
It struggles, in particular, with infeasible instances where it only solves 14 instances out of 47, needing at least 500 seconds for 12 of them.
Gurobi performed better on the feasible instances, solving 62 out of 153 instances in one hour.
On the contrary, our algorithm solves all feasible instances except one (which Gurobi does not solve as well) and all infeasible instances.
In fact, many instances are solved in less than 0.002 seconds, giving us a speed-up of several magnitudes for these instances.
More specifically, we achieve a speed-up of more than six magnitudes on almost 90\% of the infeasible instances and at least a speed-up of 100 on almost 90\% of the feasible instances.

\subsection{Optimizations}
\label{sec:eval-opts}
As seen in the previous section, our dynamic program heavily outperforms the MILP solver.
This is largely due to the proposed heuristics \emph{pruning}, \emph{greedy}, \emph{estimating}, \emph{completing}, and \emph{balancing} (see Section~\ref{sec:optim-dp-algor}).
To investigate the effect of each optimization, we use the same data set as in the previous section.
Each instance in the data set is evaluated using seven different optimization sets: all optimizations, none of them and optimization sets where one optimization is missing.
The time limit is set to $500$ seconds.

We see in Figure~\ref{fig:eval-opts}, that the dynamic program without any optimizations does not perform well, and that the pruning and the greedy heuristic have a large impact on the solvability and the running time.
In fact, the pruning heuristic solves around 85\% of the infeasible instances directly while the greedy heuristic solves around 80\% of the feasible instances directly.
Without any of these two heuristics, significantly fewer instances are solved within the given time frame.
Although the pruning heuristic does not change the solvability of feasible instances, some of them are solved more quickly.
It is also apparent that the greedy heuristic poses no significant overhead, even on infeasible instances where it never succeeds.

The other optimizations do not have such high impact in general, but contribute to a better running time in specific instances, as can be seen on the right side of Figure~\ref{fig:eval-opts}.
The balancing heuristic yields a consistent speed-up of $\frac{4}{3}$ apart from a few outliers, and the complementing heuristic solves many instances more quickly without posing an overhead on the others.

More ambiguous is estimating the success of the greedy heuristic which solves several instances faster.
However, two instances, where the estimation is too pessimistic, are solved faster without this heuristic.

\subsection{Scalability}
\label{sec:eval-scale}
We evaluate how our DP performs with a larger number of sources.
For this, we generated 300 instances as described in Section~\ref{sec:random-test-data} with six layers and between 100 and 10000 sources and scaling factors $[\log_{2\layer}(n_0) - 0.3, \log_{2\layer}(n_0) + 0.3]$ depending on the number of sources $n_0$.
We present the results in Figure~\ref{fig:eval-time-scalen}.
Similarly to previously seen results, most instances (230 out of 300) were solved directly by the pruning and the greedy heuristics.
In general, the running time of the worst case instances tends to increase for instances with more sources.
However, there is a high variance and no clear trend in the instances that are not directly solved by the pruning or the greedy heuristic.
We thus suspect that the hardness of solving an instance is not primarily dictated by its size, and we investigate this further in the next section.
\begin{figure}
  \centering
  \includegraphics[width=\columnwidth]{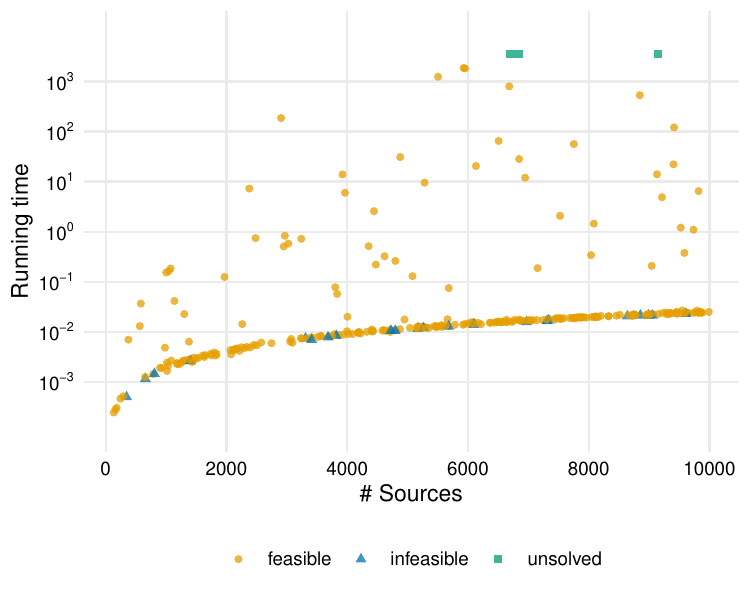}
  \caption{The running time of our DP over the number of sources.}
  \vspace*{-1em}
  \label{fig:eval-time-scalen}
\end{figure}

\subsection{Difficult Instances}
\label{sec:diff-instances}
We suspect that the ratio of the sizes of adjacent layers plays a major role for the effort our algorithm has to spend on solving an instance.
To further investigate, we generated three tests sets with different combinations of size and scaling factor.
Small instances have 1000--1500, medium instances 2000--2500, and large instances 4000--5000 sources, and we use 10 fixed scaling factors spaced equally from 1.15 to 2.5.
We generated 100 instances for each combination of size and scaling factor and solved them with a time limit of 5 minutes.
In addition, we evaluate the large instances using the dynamic program without the greedy heuristic.
The results are shown in Figure~\ref{fig:eval-factors}.
\begin{figure}
  \centering
  \includegraphics[width=\columnwidth]{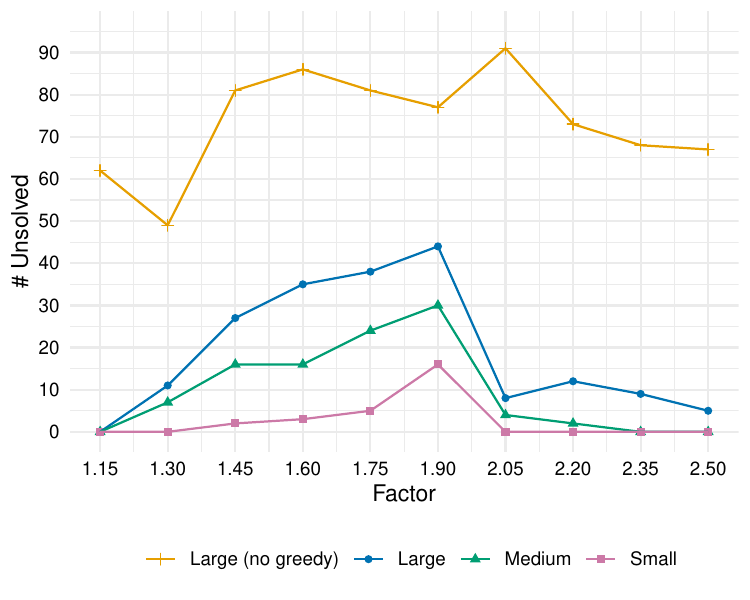}
  \caption{The number of unsolved instances after five minutes for each data set.}
  \vspace*{-1em}
  \label{fig:eval-factors}
\end{figure}

It is apparent that the scaling factor has a high impact on the hardness of an instance.
Over all instance sizes, the number of unsolved instances peaks around a scaling factor of 1.9 and decreases for both smaller and larger factors.
Clearly, larger instance sizes even increase this effect.
We suspect this is due to the decreasing capacities in each layer for smaller factors and the decreasing layer sizes for larger factors.
In both cases, there there are fewer possibilities to connect vertices which leads to fewer partial solutions and also benefits the greedy heuristic.
In fact, we observe that the greedy heuristic is most effective towards the higher and lower scaling factors.

\subsection{Optimizing \sofaclap}
\label{sec:eval-sofaclap}
As outlined before, we can obtain a feasible solution for \sofaclap by choosing an arbitrary embedding of a valid layer tree into the layer graph.
We evaluate the performance of the MILP solver with and without providing such a feasible solution found by our dynamic program.
We performed this experiment on the 200 instances from the dataset by Gritzbach~et.~al.\ \cite{gritzbach2022sofaclap}.
We measured the total time to find a solution within 0.01\% of the optimal solution.
Note that finding a feasible solution takes less than 25\,ms for each instance.

A summary of the results is given in Table~\ref{tab:eval-gurobi-dp}.
We found that, if Gurobi found a feasible solution on its own, the impact of the initial solution on the final solution quality is minuscule.
However, if Gurobi is provided a feasible solution for an instance it did not solve on its own, it could obtain a high-quality solution within the time frame.
This supports the claim that finding a feasible solution is actually the hardest part in optimizing \sofaclap.

In addition, we evaluate how our heuristic optimizations perform compared to the MILP solver by optimizing an initial solution with our local search heuristics.
We found that, on average, the heuristic solutions were 2.2 times more expensive than the ones found by Gurobi, in some cases up to three times.
However, the heuristics reached this solution quality considerably faster; on average it took Gurobi 2.5 times longer to reach the same solution quality.
The mean heuristic run took less than 27 seconds while Gurobi took 87 seconds.

\begin{table}
  \centering
  \caption{Number of instances by proximity to the optimal solution, proven infeasible or unsolved. We stopped the optimization when the cost is within 0.01\% of the lower bound found by the MILP solver.}
  \small
  \begin{tabularx}{\linewidth}{>{\ttfamily}Xrrrrrrrr}
  \toprule
  & \textless 0.01\% & \textless 1\% & \textless 2\% & Feas. & Infeas. & n.a. \\
  \midrule
  MILP    & 51       & 90    & 95    & 97       & 14         & 39 \\
  DP+MILP & 52       & 98    & 105   & 109      & 41         & 0 \\
  \bottomrule
  \end{tabularx}
  \label{tab:eval-gurobi-dp}
  \vspace{-1em}
\end{table}

\section{Conclusion}

Motivated by the \sofaclap problem, we formulated the \sofaclapfeas problem.
We designed a dynamic program with various optimizations for efficiently solving \sofaclapfeas.
Our algorithm outperforms Gurobi by several orders of magnitude and also helps finding good solutions for \sofaclap, overcoming the hurdle of finding a feasible solution.
A set of local search heuristics for \sofaclap that operate on feasible solutions complements our dynamic program.

One open question we could not answer in this work is the complexity of \sofaclapfeas.
Further investigation could also help fine-tune the optimizations used in our dynamic program and lead to an even better solution technique.
It might also be possible to extend the dynamic program itself to a heuristic solver for the \sofaclap problem by incorporating edge weights in the reconstruction process.
This would complement future work on the embedding heuristics, which still leave plenty of room for improvement.

Further research is also needed on the relation of \sofaclap and other problems.
For example, allowing the layer graph vertices to be placed arbitrarily leads to a problem resembling the geometric Steiner tree problem.
Moreover, we believe that our dynamic program can also be helpful in situations that involve not only cabling but also component cost.
Other problem classes, like hierarchical facility planning and transport problems could also benefit from insights into their relation to \sofaclapfeas.

\bibliography{sofaclap.bib}

\clearpage

\end{document}